\documentclass[letterpaper, 10 pt, conference]{ieeeconf}  

\IEEEoverridecommandlockouts                              

\usepackage[dvipsnames]{xcolor}
\usepackage{TLC}
\usepackage{wrapfig}
\usepackage{amsmath}
\usepackage{amssymb}
\usepackage{stmaryrd} 
\usepackage{pgfplotstable}
\usepgfplotslibrary{groupplots}
\usepackage{tikz}
\usetikzlibrary{plotmarks}
\usetikzlibrary{positioning}
\tikzset{
  jumpdot/.style={mark=*,solid},
  excl/.append style={jumpdot,fill=white},
  incl/.append style={jumpdot,fill=black},
}

\usepackage[autostyle=false, style=english]{csquotes}
\MakeOuterQuote{"} 

\newcommand{\Ro}{\ensuremath \R\backslash\{0\}}
\newcommand{\LT}{\ensuremath L_{2,T}}
\newcommand{\sign}{\ensuremath \operatorname{sign}}
\newcommand{\sq}[1][T]{\setlength{\fboxsep}{1pt}\raisebox{0.5pt}{\,\framebox{\scriptsize #1}}\,}
\newcommand{\argmin}{\ensuremath \operatorname*{argmin}}

\title{\LARGE \bf
Amplitude response and square wave describing functions
}

\author{Thomas Chaffey$^{1}$ and Fulvio Forni$^{2}$
    \thanks{$^{1}$ School of Electrical and Computer Engineering, University of
        Sydney, Australia. Email: \texttt{tlc37@cam.ac.uk}.}
\thanks{$^{2}$Department of Engineering, University of Cambridge,
Cambridge, CB2 1PZ, UK. Email: \texttt{f.forni@eng.cam.ac.uk}.}%
\thanks{T. Chaffey was formerly with the Department of Engineering, University of
Cambridge, where his work was supported by Pembroke College, Cambridge.}}

\begin{document}

\maketitle
\thispagestyle{empty}
\pagestyle{empty}

\begin{abstract}
    An analog of the describing function method is developed using square waves rather than sinusoids.  Static nonlinearities map square waves to square waves, and their behavior is characterized by their response to square waves of varying amplitude -- their \emph{amplitude response}.  The output of an LTI system  to a square wave input is approximated by a square wave, to give an analog of the describing function.  The classical describing function method for predicting oscillations in feedback interconnections is generalized to this square wave setting, and gives accurate predictions when oscillations are approximately square.
\end{abstract}

\section{Introduction}

The fundamental property of an LTI system which allows frequency domain analysis is that it maps a sinusoidal input to a sinusoidal output.  This allows the behavior of the system to be characterized in terms of the gain and phase shift it applies to a sinusoidal input, and gives rise to many of the foundational tools of control theory, among them transfer functions, the Nyquist diagram, Nyquist stability criterion and the Bode diagram.

As soon as a system becomes nonlinear, its response to a sinusoidal input is no longer sinusoidal, and frequency domain tools can no longer be applied directly.  Efforts to extend frequency domain analysis to nonlinear systems in a rigorous manner have played a major part in the development of absolute stability theory \cite{Liberzon2006}.  This theory relies on the separation of a system into two components connected in feedback: one component which is LTI, and therefore lends itself to frequency domain analysis, and another component containing any nonlinearities and other troublesome elements.  Such a separation is called a Lur'e system.  

A fruitful heuristic approach to study Lur'e systems is \emph{describing function analysis} \cite{Krylov1950}.  The troublesome nonlinear component is replaced by a component which maps sinusoids to sinusoids, by approximating its output by a sinusoid using a least-squares fitting.  Frequency domain tools can then be applied to the approximated feedback interconnection, and instability in the approximated system is often a good predictor of oscillation in the true system \cite{Slotine1991}.  Despite being approximate, the describing function method is a widely used method for nonlinear control design in practice \cite{Slotine1991, vandenEijnden2020a}.

The standard describing function approach relies on the output of the nonlinear component being approximately sinusoidal, so that its approximation by a sinusoid is reasonably accurate.  The accuracy of the approach can be quantified explicitly \cite{Bergen1971, Mees1975}, and there have been many efforts to extend describing function analysis to wider classes of systems and improve its accuracy, for example by allowing multiple inputs \cite{Gelb1968}, Gaussian inputs \cite{Booton1954}, generalized outputs \cite{Bertoni1969, Jopling1964}, incorporating higher order terms \cite{Yang1994, Nuij2006} and fitting the sinusoidal approximation with alternatives to least-squares \cite{Gibson1963, Prince1954}.  However, all of these extensions retain the basic philosophy that the nonlinear component's output should be approximately sinusoidal.

The question which motivates this paper is whether there exist other classes of signals and systems such that the signal class is preserved by the system, and an analogous approach to frequency domain analysis might be possible.  We begin with the observation that a static nonlinear function maps a square wave to another square wave, and proceed to develop an upside-down version of the describing function method: static nonlinear functions are the ``easy'' components, whose behavior is characterized by their response to square waves of varying amplitude; LTI systems are the ``hard'' components, whose output to a square wave input must be approximated by a square wave.  We develop these ideas into a method of predicting oscillations in a feedback interconnection which are approximately square, rather than approximately sinusoidal as predicted by the classical describing function method.  Square wave oscillations are common in electronic applications such as relaxation oscillators \cite{Slotine1991} and DC/DC power converters \cite{Sanders1993}.

The remainder of this paper is structured as follows.  We briefly introduce notation in Section~\ref{sec:prelims}, before giving a summary of the classical describing function method in Section~\ref{sec:recap}.  In Section~\ref{sec:square}, we develop an analog of the frequency response for static nonlinearities, which we call the \emph{amplitude response}, using square waves as inputs, and conclude the section with a square wave version of the Nyquist criterion.  Finally, in Section~\ref{sec:adf}, we introduce the amplitude describing function of an LTI system, and a method for predicting the existence of approximately square oscillations in feedback systems. Examples are given in Section~\ref{sec:examples}, and conclusions are drawn in Section~\ref{sec:conclusions}.

\section{Preliminaries and notation}\label{sec:prelims}

We let $\mathcal{B}$ denote the space of pointwise bounded signals $u: \R \to \R$.
We let $L_2(\mathbb{T}, \mathbb{F})$ denote the space of square integrable functions
mapping $\mathbb{T} \to \mathbb{F}$, where $\mathbb{F}$ is $\mathbb{C}$ or $\mathbb{R}$, and use the shorthand notation $\LT$ to denote $L_2([0,
T], \R)$.  These spaces are equipped with the standard inner product:
\begin{IEEEeqnarray*}{rCl}
\ip{u}{y} := \int_\mathbb{T} u(t)^* y(t) \dd{t},
\end{IEEEeqnarray*}
where $z^*$ denotes the conjugate transpose of $z$.
A signal $u \in \mathcal{B}$ is said to be \emph{periodic} if there exists some $T > 0$
such that $u(t) = u(t + T)$ for all $t \in \R$.
Periodic signals may be considered to belong to $\LT$, by restricting to a single
period.  Given a scalar $\tau \in [0, T]$, we define the
\emph{periodic $\tau$-delay} $P_\tau: \LT \to \LT$ by
\begin{IEEEeqnarray*}{rCl}
P_\tau u (t) = \begin{cases}
        u(t - \tau + T) & 0 \leq t < \tau\\
        u(t - \tau) & \tau \leq t \leq T.
\end{cases}
\end{IEEEeqnarray*}

An operator $N: \mathcal{B} \to \mathcal{B}$ is said to be \emph{periodicity
preserving} if $u \in \LT \implies N(u) \in \LT$ for all $T > 0$, in which case we can restrict to a single period to induce an operator on $\LT$.
We make the standing assumption that any operator $N: \mathcal{B} \to \mathcal{B}$ satisfies $N(\mathbf{0}) = \mathbf{0}$, where
$\mathbf{0}$ denotes the zero signal $u(t) = 0$ for all $t$.
Where there is no risk of ambiguity, we use $G$ to denote both an LTI operator and its transfer function representation. 
A similar approach is adopted for generic (nonlinear) operators:
$N$ denotes both the operator and its representation.

\section{The classical describing function method}\label{sec:recap}

We begin with a brief summary of the classical describing function method, which predicts the existence of self-sustaining oscillations in the negative feedback interconnection of an LTI system and a nonlinear operator, as shown in Figure~\ref{fig:df_fb}.  In general, proving the existence of such oscillations is a hard problem, and the describing function method is a heuristic, which tests for the existence of purely \emph{sinusoidal} oscillations in an approximation of the original system.  Given a sinusoidal input, the LTI system $G$ produces a sinusoidal output. For the nonlinear operator $N$, however, this is no longer the case: a sinusoidal input produces an output with, in general, many sinusoidal components at different frequencies. The describing function method approximates this output by a single frequency component, and searches for a purely sinusoidal oscillation in the feedback connection of this approximation with $G$.

\begin{figure}
    \centering
    \includegraphics{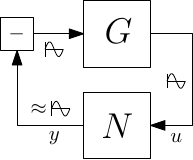}
    \caption{The describing function method predicts oscillations in the negative feedback interconnection of an LTI system $G$ with a nonlinear operator $N$, by approximating the output $y$ of the nonlinear element by a sinusoid.}
    \label{fig:df_fb}
\end{figure}

If $u \in \LT$ is defined by $u(t) = \alpha \cos(\omega t)$, $\omega =
(2\pi)/T$, and $y = N(u)$ for some nonlinear operator $N: \LT \to \LT$, we can approximate $y(t)$ by a sinusoid, by solving the following least-squares problem:
\begin{IEEEeqnarray*}{rCl}
    y(t) &\approx& \beta \cos(\omega t + \phi)\\
    (\beta, \phi) &=& \argmin_{\tilde \beta, \tilde \phi} \norm{\tilde \beta \cos(\omega  \cdot + \tilde \phi) - y(\cdot)}^2.\IEEEyesnumber \label{eq:approx}
\end{IEEEeqnarray*}

Expressing $u(t)$ as $(\alpha/2)(e^{j\omega t} + e^{-j\omega t})$, we have the map
\begin{IEEEeqnarray*}{rCl}
     \frac{\alpha}{2} e^{j\omega t} + \frac{\alpha}{2} e^{-j\omega t}  &\mapsto& \frac{\beta}{2} e^{j\phi} e^{j\omega t} + \frac{\beta}{2} e^{-j\phi} e^{-j\omega t},
\end{IEEEeqnarray*}
from which we obtain the transfer function

\begin{IEEEeqnarray*}{rCl}
    \frac{\alpha}{2} e^{j\omega t}  &\mapsto& \tilde N(\omega, \alpha) \frac{\alpha}{2} e^{j\omega t}\\
    \tilde N(\omega, \alpha) &:=&  \frac{\beta}{\alpha} e^{j\phi}.
\end{IEEEeqnarray*}

The function $\tilde N(\omega, \alpha)$ is called the \emph{describing function of
$N$}, and defines a mapping from sinusoids to sinusoids, illustrated in
Figure~\ref{fig:df}. 

\begin{figure}[hb]
    \centering
    \includegraphics[width=\linewidth]{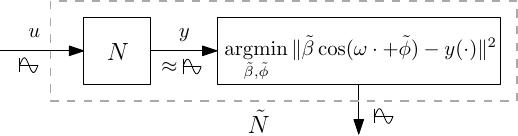}
    \caption{The describing function of $N$, denoted by $\tilde N$, defines a mapping
    from sinusoids to sinusoids by approximating the output of $N$ by a sinusoid.}
    \label{fig:df}
\end{figure}

We now have two operators which map a sinusoidal input to a sinusoidal output,
represented by the complex numbers $G(j\omega)$ and $\bar N(\omega, \alpha)$.  We can therefore solve for a sinusoidal solution 
$w(t) = \alpha e^{j\omega t}$ to the feedback equations
\begin{IEEEeqnarray*}{rCl}
    u(t) &=& G(j\omega) w(t)\\
    y(t) &=& \tilde N(\omega, \alpha) u(t)\\
    w(t) &=& - y(t).
\end{IEEEeqnarray*}
Simplifying these equations, we obtain the condition $G(j\omega)\tilde N(\omega,
\alpha) = -1$, which can be tested for graphically by plotting the loci $G(j\omega)$
and $-1/\tilde N(\omega, \alpha)$ on the complex plane and finding an intersection.
If $N$ is a static nonlinearity, this test is simplified, as $\tilde
N(\omega, \alpha)$ is then constant with respect to $\omega$, so $\tilde N(\alpha)$ can be plotted as a function of $\alpha$ alone.

The describing function is closely related to the Fourier series of the signal $y(t)$, which is the series expansion
\begin{IEEEeqnarray}{rCl}
    y(t) &=& \sum_{n = -\infty}^\infty c_n e^{-j2\pi nt/T},\label{eq:Fourier_series}
\end{IEEEeqnarray}
where the coefficients are given by
\begin{IEEEeqnarray}{rCl}
    c_n := \frac{1}{T} \int_0^T y(t) e^{-j2\pi nt/T} \dd{t}. \label{eq:Fourier_coefficients}
\end{IEEEeqnarray}
The following proposition states that the describing function approximation is equivalent to truncating the Fourier series of $y$ to the
$1$ and $-1$ terms.

\begin{proposition}\label{prop:Fourier}
    The approximation \eqref{eq:approx} is equal to the approximation
\begin{IEEEeqnarray}{rCl}
    y(t) &\approx& c_1 e^{-j2\pi t/T} + c_{-1} e^{-j2\pi
    t/T}.\label{eq:Fourier_approx}
\end{IEEEeqnarray}
\end{proposition}

\begin{proof}
It follows directly from \eqref{eq:Fourier_series} and
\eqref{eq:Fourier_coefficients} that \eqref{eq:Fourier_approx} is the projection of $y(t)$ onto the subspace of $L_2(\mathbb{C}, \R)$ spanned by $(e^{j\omega t}, e^{-j\omega t})$.  It follows from the projection theorem \cite[Thm. 1]{Luenberger1969} that 
    \begin{IEEEeqnarray*}{rCl}
        (c_1, c_{-1}) &=& \argmin_{\tilde c_1, \tilde c_{-1} \in \C} \norm{\tilde c_1 e^{j\omega \cdot} + \tilde c_{-1} e^{-j\omega \cdot} - y(\cdot)}^2.
    \end{IEEEeqnarray*}
    We also note from \eqref{eq:Fourier_coefficients} that $c_{-1} = c_1^*$, so the minimization reduces to
    \begin{IEEEeqnarray*}{rCl}
        c_1 &=& \argmin_{\tilde c_1 \in \C} \norm{\tilde c_1 e^{j\omega \cdot} + \tilde{c}_{1}^* e^{-j\omega \cdot} - y(\cdot)}^2.
    \end{IEEEeqnarray*}
    Expressing $c_1 = \beta e^{j\phi}$, this further simplifies to give
    \eqref{eq:approx}:
    \begin{IEEEeqnarray*}{+rCl+x*}
        (\beta, \phi) &=& \argmin_{\tilde \beta e^{j\tilde \phi} \in \C} \norm{\tilde \beta \cos(\omega \cdot + \tilde \phi) - y(\cdot)}^2.& \qedhere
    \end{IEEEeqnarray*}
\end{proof}

\section{The amplitude response of a nonlinearity}\label{sec:square}

In this paper, we revisit the describing function method by using a square wave 
approximation in place of the sinusoidal approximation \eqref{eq:approx}.  First,
however, we examine the class of systems which map square waves to square waves:
these take the role of the LTI system in the classical describing function setting of
Figure~\ref{fig:df_fb}.

The set of square waves is characterized by a base signal of period $T$, defined as follows: 
\begin{IEEEeqnarray*}{rCl}
\sq(t) &:=& \begin{cases} \frac{1}{T} & 0 \leq t < \frac{T}{2}\\
    -\frac{1}{T} & \frac{T}{2} \leq t < T.
                       \end{cases}
\end{IEEEeqnarray*}
Analysis will be performed for square waves of a given period $T$.  The set of square
waves of a given period is formed by changing the amplitude and phase of $\sq$.
We denote the phase delayed signal $P_\tau \sq$ by $\sq_\tau$, for $\tau \in [0, T]$.
The signal $\sq_\tau$ is illustrated in Figure~\ref{fig:square}.  Note that
$\norm{\sq_\tau} = 1$ and $\int_0^T \sq_\tau(t) \dd{t} = 0$ for all $T, \tau$.
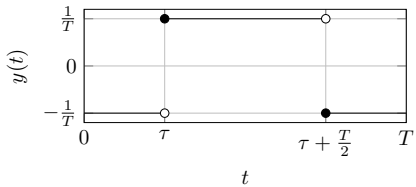
\begin{figure}[hb]
    \centering
\begin{tikzpicture}[scale=0.8]
    \begin{axis}[
        xlabel={$t$},
        ylabel={$y(t)$},
        height=0.4\linewidth,
        width=0.8\linewidth,
        grid=both,
        xtick={0, 0.25, 0.75, 1},
        xticklabels={$0$, $\tau$, $\tau + \frac{T}{2}$, $T$},
        ytick={-1, 0, 1},
        yticklabels={$-\frac{1}{T}$, $0$, $\frac{1}{T}$},
        xmax=1.0, xmin=0.0]
        
    \addplot[domain=0:0.25, samples=2, black] {-1}; 
    \addplot[excl] coordinates {(0.25, -1)};
    \addplot[domain=0.25:0.75, samples=2, black] {1}; 
    \addplot[incl] coordinates {(0.25, 1)};
    \addplot[excl] coordinates {(0.75, 1)};
    \addplot[domain=0.75:1.0, samples=2, black] {-1}; 
    \addplot[incl] coordinates {(0.75, -1)};
    \end{axis}
\end{tikzpicture}
\caption{The signal $\sq_\tau$.}\label{fig:square}
\end{figure}

The crucial property of LTI systems that enables frequency domain analysis is that
sinusoids are eigenfunctions: an LTI system  maps a sinusoid to
another sinusoid, with a gain and phase shift depending on the frequency of the
input.  In this paper, we work with systems which enjoy a similar property, with
respect to the set of square waves.  This is formalized as follows.

\begin{definition}\label{def:square_preserving}
We say that a periodicity preserving
operator $\mathcal{N}: \mathcal{B} \to \mathcal{B}$ is \emph{square-preserving} if, for all
$\alpha \in \R$, $T > 0$, and $\tau \in [0, T]$, there exists $N(\alpha,T) \in \mathbb{C}$  such that
\begin{equation*}
\mathcal{N}(\alpha \sq_\tau) = |N(\alpha, T)||\alpha|
\sq_{\tau + \frac{T}{2\pi}\angle N(\alpha, T)}.\qedhere
\end{equation*}
\end{definition}

The first example of a square-preserving system is an odd static nonlinearity, as shown in the following theorem.

\begin{theorem}\label{thm:static_NL}
    Let $\Phi$ be a static nonlinearity which is odd, that is, $-\Phi(u) = \Phi(-u)$.
    Then
    \begin{IEEEeqnarray*}{rCl}
    \Phi(\alpha\sq_\tau) = |\Phi(\alpha)|\sq_{\tau+(T/4)(1- \sign(\Phi(\alpha)))}.
    \end{IEEEeqnarray*}
\end{theorem}
\begin{proof}
We show the case $\tau=0$. The case $\tau \neq 0$ is similar.
    Applying $\Phi$ to $\alpha\sq$ gives
    \begin{IEEEeqnarray*}{+rCl+x*}
    \Phi(\alpha\sq) &=& \begin{cases} \Phi(\alpha) & 0 \leq t < \frac{T}{2}\\
                           \Phi(-\alpha) & \frac{T}{2} \leq t < T
                       \end{cases}&\\
               &=& \begin{cases} \Phi(\alpha) & 0 \leq t < \frac{T}{2}\\
                           -\Phi(\alpha) & \frac{T}{2} \leq t < T
                       \end{cases}&\\
               &=& \Phi(\alpha)\sq&\\
   &=& |\Phi(\alpha)|\sq_{(T/4)(1- \sign(\Phi(\alpha) )}. & \qedhere
    \end{IEEEeqnarray*}
\end{proof}

The Nyquist diagram of an LTI system plots output gain and phase shift as a function of the frequency of an input sinusoid.  In a similar manner, we can input a square wave to a square-preserving nonlinearity and plot the gain and phase shift of the output, to produce an \emph{amplitude Nyquist diagram}.  We define this formally as follows.

\begin{definition}
    Let $N: \mathcal{B} \to \mathcal{B}$ be a square-preserving operator. Then the \emph{amplitude Nyquist diagram} of $N$ at period $T > 0$ is the region of the complex plane defined by
    \begin{IEEEeqnarray*}{+rCl+x*}
         \nyq{N} &:=& \left\{|N(\alpha, T)|e^{j\angle N(\alpha, T)}\;\middle| \; \alpha \in \R \right\}.& \qedhere
    \end{IEEEeqnarray*}
\end{definition}

We can also plot an amplitude Bode diagram, by plotting the gain and phase of the amplitude Nyquist diagram as functions of $\alpha$. 

\begin{example}
    Consider $y(t) = \sat{ku(t)}$, where $k>0$ and
    \begin{IEEEeqnarray}{rCl}
    \sat{x} &=& \begin{cases} x & |x| \leq 1\\
        |x|/x & \text{ otherwise.}
    \end{cases}\label{eq:sat}
    \end{IEEEeqnarray}
    This operator is square-preserving by Theorem~\ref{thm:static_NL}.
    Amplitude Nyquist and Bode diagrams for this operator are shown in
    Figure~\ref{fig:sat_nyqa}.  In this case, the amplitude Nyquist diagram turns out to be identical to the Nyquist diagram of the operator's regular describing function, although this is in general not the case.  The saturation is \emph{low-pass} in an amplitude
    sense, allowing low amplitude signals to pass unattenuated, but attenuating large
    amplitude signals.  In the limit $k \to \infty$, we obtain the ideal relay, $y(t)
    = \sign(u(t))$, which may be thought of as the amplitude equivalent of an
    integrator, or ideal low-pass element.
\end{example}

\begin{example}\label{ex:delay}
    The amplitude-dependent delay $D: \mathcal{B} \to \mathcal{B}$, defined by
    \begin{align*}
        u(t) &\mapsto u(t - \gamma(t)),\\
        \gamma(t) &= \max_{\nu < t}|u(\nu)|,
    \end{align*} 
    is square-preserving, mapping $\alpha \sq$ to $\alpha \sq_{|\alpha|}$.  Its
    amplitude Nyquist diagram is therefore the unit circle, for any $T$.
\end{example}

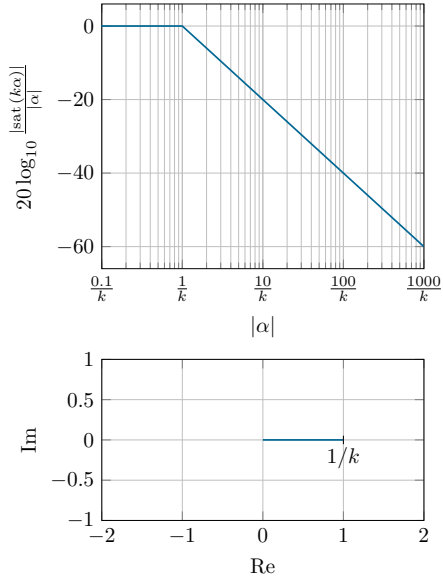
\begin{figure}[h]
    \centering
\begin{tikzpicture}[
    declare function={
        sat(\x) = (\x <= 1) * 1 + (\x > 1) * (1/\x);
    }, scale=0.8
    ]
\begin{groupplot}[
  group style={
    group size=1 by 2,
    y descriptions at=edge left,
    vertical sep=15mm
}, width=0.8\linewidth]
\nextgroupplot[
    domain=0:5,
    ylabel=$20\log_{10}\frac{|\sat{k \alpha}|}{|\alpha|}$,
    xlabel=$|\alpha|$,
    grid=both,
    xtick={0.01, 0.1, 1, 10, 100, 1000},
    xticklabels={$\frac{0.01}{k}$, $\frac{0.1}{k}$, $\frac{1}{k}$, $\frac{10}{k}$, $\frac{100}{k}$, $\frac{1000}{k}$},
    xmin=0.1, xmax=1000,
    xmode=log
]
\addplot[MidnightBlue, thick, no marks, samples at={0.01, 0.1, 1, 10, 100, 1000}] {20*log10(sat(x))};

\nextgroupplot[
        ylabel={$\Im$},
        xlabel={$\Re$},
        grid=both,
        xmin=-2.0, xmax=2.0,
        ymin=-1.0, ymax=1.0,
        legend pos=south east, 
        axis equal image
    ]

    \node [below] at (axis cs:  1, 0) {$1/k$};
    \addplot[only marks, mark=|] coordinates {
        (1, 0)}; 
    
    \addplot[MidnightBlue, thick, no marks] coordinates {(0, 0) (1, 0)};

\end{groupplot}
\end{tikzpicture}
\caption{
Amplitude Bode diagram (above, gain only) and amplitude Nyquist diagram
(below) of $\sat{ku(\cdot)}$ for $k > 0$.}

\label{fig:sat_nyqa}
\end{figure}

We conclude this section by giving a theorem, reminiscent of the classical Nyquist
criterion, that allows us to establish properties of the negative feedback interconnection
\begin{IEEEeqnarray}{rCl}
    e &=& r - y \label{eq:unity_gain_1}\\
    y &=& N(e)\label{eq:unity_gain_2}
\end{IEEEeqnarray}
from the amplitude Nyquist diagram of $N$.

\begin{theorem}
Let $N: \mathcal{B} \to \mathcal{B}$ be square-preserving and fix $T > 0$.  If the
amplitude Nyquist diagram of $N$ at period $T$ contains the point $-1$ for a value of
$\alpha \neq 0$, the feedback loop defined by
\eqref{eq:unity_gain_1}--\eqref{eq:unity_gain_2} admits a self-sustaining square wave
oscillation.
\end{theorem}

\begin{proof}
    Set $r = \mathbf{0}$.  We have $-1 \in \nyq{N}$,
    so there exists $\alpha \neq 0$ such that, setting $e = \alpha \sq$, we have $y =
    |N(\alpha, T)||\alpha|\sq_{\tau}$ with
    \begin{IEEEeqnarray*}{rCl}
        |N(\alpha, T)| = 1, \quad
        \tau &=& \frac{T}{2} + kT, \quad k \in \mathbb{Z}.
        \end{IEEEeqnarray*}
    This implies $y = -e$.  Substituting in \eqref{eq:unity_gain_1} gives $r =
    \mathbf{0}$.  Therefore, \eqref{eq:unity_gain_1}--\eqref{eq:unity_gain_2} admit
    at square wave solution with $r = \mathbf{0}$. 
\end{proof}

\section{Amplitude describing functions}\label{sec:adf}

In Section~\ref{sec:square}, we examined systems which map square
waves to square waves, and the stability of such systems in feedback.  In this
section, we look at systems which do not preserve square waves, and develop an
analog of the classical describing function which approximates such systems by
square-preserving systems.  We then develop a graphical criterion for predicting the
existence of square wave oscillations in feedback interconnections.

\subsection{The amplitude describing function of an LTI system}

Given a period-preserving operator $N: \mathcal{B} \to \mathcal{B}$ which is not
square-preserving, we can approximate $N$ by a square-preserving operator called the
\emph{amplitude describing function} of $N$, by approximating $y = N(\alpha \sq)$ by
a square wave using a least-squares fitting, as illustrated in Figure~\ref{fig:adf}.  We formalize this as follows.

\begin{figure}[hb]
    \centering
    \includegraphics{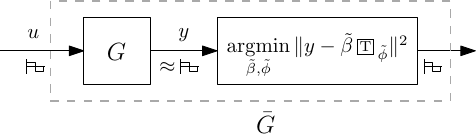}
    \caption{The amplitude describing function of $G$, denoted by $\bar G$, defines a
        square-preserving mapping
    approximating the output of $G$ by a square wave.}
    \label{fig:adf}
\end{figure}

\begin{definition}\label{def:adf}
     Given an operator $N: \LT \to \LT$, its \emph{amplitude describing function} is the operator $\bar N(T, \alpha): \LT \to \LT$ defined by
     \begin{IEEEeqnarray*}{+rCl+x*}
        \bar N(T, \alpha)(\alpha \sq) &:=& \beta \alpha \sq_{\frac{T}{2\pi}\phi}&\\
        (\beta, \phi) &:=& \argmin_{\tilde \beta \geq 0, \tilde \phi \in [0, 2\pi]} \norm{N(\alpha \sq) - \tilde \beta\alpha \sq_{\tilde \phi}}^2.&\qedhere
     \end{IEEEeqnarray*}
\end{definition}
As with the classical describing function, the amplitude describing function can be encoded as a complex number $\beta e^{j\phi}$, and we will abuse notation and use $\bar N(T, \alpha)$ to denote both the describing function operator and its complex number representation.
Proposition~\ref{prop:Fourier} characterizes the regular describing function in terms
of the Fourier transform.  The following theorem gives an
analogous result for the amplitude describing function.
\begin{theorem}\label{thm:square_transform}
     Given an operator $N: \LT \to \LT$, $\alpha \in \Ro$, $T > 0$, let $y = N(\alpha \sq)$.  Then 
     $\bar N(T, \alpha) = \beta e^{j\phi}$, where
      \begin{IEEEeqnarray*}{rCl}
        \phi &=& \argmin\limits_{\tilde{\phi} \in [0, \pi]} 
            -\ip{y}{\sq_{\frac{T}{2\pi}\tilde\phi}}^2 \IEEEyesnumber\label{eq:opt}\\
        \beta &=& \frac{1}{\alpha}\ip{y}{\sq_{\frac{T}{2\pi}\phi}}.
     \end{IEEEeqnarray*}    
\end{theorem}
\begin{proof}
 
    Define $\tau = \frac{T}{2\pi}\phi$ and note that 
    $\phi \in [0,\pi]$ implies $\tau\in [0, T/2]$.  We further note that, for $\tau >
    T/2$, $\sq_{\tau} = -\sq_{\tau - T/2}$, allowing us to restrict $\tau$ to $[0,
    T/2]$ by allowing $\beta$ to be negative. For each $\tau \in [0, T]$,
    $\sq_{\tau}$ has unit norm and spans a one-dimensional subspace in $\LT$.  It then follows from the projection theorem \cite[Thm. 1]{Luenberger1969} that
    \begin{equation*}
        \argmin_{\tilde \beta} \norm{y - \tilde \beta\alpha \sq_{\tau}}^2 
        = \frac{1}{\alpha}\ip{y}{\sq_{\tau}}, 
    \end{equation*}
    allowing us to eliminate $\beta$ from the minimization in Definition~\ref{def:adf} to obtain
    \begin{IEEEeqnarray*}{+rCl+x*}
        &&\min_{\tau} \norm{y - \frac{1}{\alpha}\ip{y}{\sq_{\tau}}\alpha\sq_{\tau}}^2&\\
        &=&\min_{\tau} \norm{y}^2 + \norm{\ip{y}{\sq_{\tau}}\sq_{\tau}}^2 - 2\ip{y}{\ip{y}{\sq_{\tau}}\sq_{\tau}}&\\
        &=&\min_{\tau} \ip{y}{\sq_{\tau}}^2 - 2\ip{y}{\sq_{\tau}}^2&\\
        &=& \min_{\tau} -\ip{y}{\sq_{\tau}}^2.&\qedhere
    \end{IEEEeqnarray*}
\end{proof}

The minimization \eqref{eq:opt} appears expensive to compute, however Proposition~\ref{thm:transform_simple} in the appendix shows that the cost only has to be evalutated at a small number of test phase shifts.
This approximation is closely related to, but distinct from, the Haar wavelet transform with a periodic boundary condition \cite[$\S
7.5.1$]{Mallat2009} the square wave transform of \cite{Pender1992}, and the matched
filter \cite{Turin1960}.

In general, the amplitude describing function depends on both the input's amplitude
and period.  In the case of a linear system, however, the amplitude describing
function becomes independent of amplitude: $\bar G(T, \alpha_1) = \bar G(T,
\alpha_2)$ for all $T > 0$ and $\alpha_1, \alpha_2 \in \R\backslash \{0\}$.  This
follows directly from linearity.

%

\subsection{Harmonic Balance and the Extended Nyquist Criterion}

In this section, we describe a square wave analog of the classical describing
function method for predicting oscillations in the negative feedback interconnection
of an LTI system with an odd static nonlinearity, illustrated in Figure~\ref{fig:adf_fb} and
defined as follows:
\begin{IEEEeqnarray}{rCl}
    e(t) &=& -y(t)\label{eq:true_balance}\\
    y(t) &=& G(\Phi(e(t))).\label{eq:true_feedback}
\end{IEEEeqnarray}
We replace prediction of an oscillation in the true system, in general a hard
problem, with a heuristic: prediction of an oscillation in an approximate system.  The odd static nonlinearity $\Phi$ is square preserving by Theorem~\ref{thm:static_NL}.  We form an approximate system by replacing the LTI component with its amplitude
describing function.  
Both components in the system then map square waves to square waves, and we can write the feedback equations as follows:
\begin{IEEEeqnarray}{rCl}
    e(t) &=& -y(t)\label{eq:harmonic_balance}\\
    y(t) &=& \bar{G}(T)(\Phi(e(t))).\label{eq:feedback}
\end{IEEEeqnarray}
Fixing $T > 0$ and
letting $e(t) = \alpha \sq$ for some $\alpha \in \R\backslash\{0\}$, we can rewrite
the feedback equations to identify a condition for the existence of a self-sustaining
oscillation of period $T$.  We first note that $\Phi(\alpha\sq)(t) = \gamma(\alpha) \alpha
\sq_{\frac{T}{2\pi}\theta(\alpha)}(t)$ for some $\gamma(\alpha) > 0, \theta(\alpha) \in [0, 2\pi]$.
Likewise, $\bar{G}(T)(\gamma(\alpha) \alpha \sq_{\frac{T}{2\pi}\theta(\alpha)}) = \beta(T)
\gamma(\alpha) \alpha \sq_{\frac{T}{2\pi}(\phi(T) + \theta(\alpha))}$ for some $\beta(T) > 0, \phi(T) \in [0, 2\pi]$.  The harmonic balance condition \eqref{eq:harmonic_balance} then becomes
\begin{IEEEeqnarray}{rCl}
    \beta(T)\gamma(\alpha) = 1, \ \phi(T) + \theta(\alpha) = \pi + 2k\pi, \ k \in \mathbb{Z}.\label{eq:harmonic_balance_square}
\end{IEEEeqnarray}
Condition~\eqref{eq:harmonic_balance_square} can be tested graphically, in a method analogous to the extended Nyquist criterion used in standard describing function analysis.

\begin{figure}
    \centering
    \includegraphics{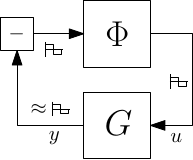}
    \caption{The amplitude describing function method predicts oscillations in the negative feedback interconnection of an LTI system $G$ with a nonlinear operator $\Phi$, but signals are approximated by square waves rather than sinusoids, and it is the output of the LTI component $G$ that must be approximated.}
    \label{fig:adf_fb}
\end{figure}

\begin{theorem}\label{thm:extended_Nyquist}
    Given $T > 0$, an LTI operator $G$ and an odd static nonlinearity $\Phi$, the feedback
    equations~\eqref{eq:harmonic_balance}--\eqref{eq:feedback} admit a solution $e(t)
    = \alpha \sq(t)$ for some $\alpha \in \R\backslash\{0\}$ if and only if 
    \begin{equation}\label{eq:extended_Nyquist}
    \nyq{\bar{G}(T)} \cap \frac{-1}{\nyq{\Phi}}  \neq \varnothing,
    \end{equation}
    where the operation $-1/\cdot$ is performed elementwise on the set $\nyq{\Phi}$.
\end{theorem}
\begin{proof}
    We begin by noting that
    \begin{IEEEeqnarray}{rCl}
    \nyq{\bar G(T) \Phi} &=& \nyq{\bar G(T)} \nyq{\Phi}. \label{eq:product}
    \end{IEEEeqnarray}
    Indeed, fixing $\alpha \in \R\backslash\{0\}$, we have already noted above that
    $(\bar G (T) \Phi)(\alpha \sq) = \beta(T)\gamma(\alpha)\alpha\sq_{\frac{T}{2\pi}(\phi(T) +
    \theta(\alpha))}$.  The corresponding point on $\nyq{\bar G(T) \Phi}$ is therefore
    \begin{IEEEeqnarray*}{rCl}
        \beta(T)\gamma(\alpha)e^{j(\phi(T) + \theta(\alpha))} = \beta(T) e^{j\phi(T)}
        \gamma(\alpha) e^{j\theta(\alpha)},
    \end{IEEEeqnarray*}
    which is the product of the points corresponding to $\alpha$ in $\nyq{\bar G(T)}$
    (which is a single point, independent of $\alpha$) and $\nyq{\Phi}$.  Since
    $\alpha$ was chosen arbitrarily, \eqref{eq:product} follows.

    We now have the following chain of equivalences:
    \begin{IEEEeqnarray*}{+rCl+x*}
    \eqref{eq:extended_Nyquist} &\iff& -1 \in \nyq{\bar G(T)} \nyq{\Phi}&\\
                                &\iff& -1 \in \nyq{\bar G(T) \Phi}&\\
                                &\iff& \eqref{eq:harmonic_balance_square}.&\qedhere
    \end{IEEEeqnarray*}
\end{proof}

Theorem~\ref{thm:extended_Nyquist} gives a necessary and sufficient condition for the
existence of a square wave solution to
\eqref{eq:harmonic_balance}--\eqref{eq:feedback} \emph{of a particular period $T$}.
To search for a square wave solution of arbitrary period, we apply
Theorem~\ref{thm:extended_Nyquist} over a range of periods $T$.  Since $\nyq{\Phi}$
is independent of $T$, and $\nyq{\bar G (T)}$ is independent of $\alpha$ (and, for
each $T$, contains a single point), this amounts to plotting the locus $\nyq{\bar
G(T)}$ as a function of $T$ on the same axes as the locus $\nyq{\Phi}$ as a function
of $\alpha$, and searching for an intersection.  The value of $T$ which gives the
point $\nyq{\bar G(T)}$ at the intersection is the period of the square wave
solution, and the value of $\alpha$ which give the point in $\nyq{\Phi}$ is the amplitude of the square wave solution, measured at the output of $\bar G$.

In Theorem~\ref{thm:extended_Nyquist}, we have restricted ourselves to the case where $\Phi$ is static, meaning $\nyq{\Phi}$ is equal for all $T$.  The theorem extends in a straightforward manner to the case where this assumption is not met, and $\Phi$ is an arbitrary square-preserving nonlinearity.  In this case, \eqref{eq:harmonic_balance_square} becomes 
\begin{IEEEeqnarray*}{rCl}
    \beta(T)\gamma(T, \alpha) = 1, \  \phi(T) + \theta(T, \alpha) = \pi + 2k\pi, \ k \in \mathbb{Z},
\end{IEEEeqnarray*}
and we must plot the amplitude Nyquist diagram of $\Phi$ over both $T$ and $\alpha$ to find an intersection which satisfies these conditions.  This is demonstrated in Example~\ref{ex:sat_delay}.  
This more complicated case is analogous to the classical describing function method applied to nonlinearities with memory, in which case $\tilde N$ is a function both of frequency and amplitude.

Theorem~\ref{thm:extended_Nyquist} allows us to determine the existence, period and amplitude of oscillations in the system \eqref{eq:harmonic_balance}--\eqref{eq:feedback}.  However, this system is only an approximation of the true system, and the question remains whether an oscillatory solution to the approximate system is a good predictor of an oscillation in the true system.  For the classical describing function method, the prediction is reliable if the LTI system is sufficiently low-pass: the LTI system then filters out any higher order harmonics which have been neglected in the describing function approximation of the nonlinearity (an argument which is usually made intuitively \cite{Slotine1991}, but can be made precise \cite{Bergen1971, Mees1975}).  For amplitude describing functions, we conjecture that the dual is true: the prediction is reliable if the static nonlinearity is sufficiently low-pass in an amplitude sense, in which case the nonlinearity filters out any high amplitude components which have been neglected in the square wave approximation of the LTI system's output.  This is left as a conjecture, but demonstrated empirically in the following section.

\section{Examples}\label{sec:examples}

\begin{example}\label{ex_1}

In this example, we demonstrate the amplitude describing function method on a third order Goodwin oscillator \cite{Goodwin1965}, and compare its accuracy against the regular describing function method as the nonlinearity becomes increasingly low-pass in an amplitude sense.  The dynamics of the system are described by
\begin{IEEEeqnarray*}{rCl}\label{eq:example_1}
    y(s) &=& \frac{1}{(s+1)^3} e(s)\IEEEyesnumber\\
    u(t) &=& \sat{ky(t)},\\
    e(t) &=& - u(t),
\end{IEEEeqnarray*}
where $\operatorname{sat}$ is defined in \eqref{eq:sat}
and $k > 0$ is a variable gain which determines the steepness of the saturation.  The plots for the regular and amplitude describing function analyses are shown in Figure~\ref{fig:example_1}\footnote{Code for plotting the adf and Nyqa can be found at \texttt{https://github.com/ThomasChaffey/square-wave-\\describing-functions}}. Experimentally, a self-sustaining oscillation is present in the system for $k \geq 8$.  This is also the range predicted by the regular describing function method.  The amplitude describing function method predicts oscillations for $k \geq 9.53$.  Figure~\ref{fig:example_1_period} shows the period of oscillation as a function of $k$, as well as the period predicted by the regular and amplitude describing function methods.  For low $k$, the regular describing function method performs well.  As $k$ becomes larger, the saturation becomes increasingly low-pass in an amplitude sense, and the oscillation period predicted by the amplitude describing function becomes increasingly accurate.  Example oscillations are plotted for a range of $k$ values in Figure~\ref{fig:ex_1_oscillations}.
\end{example}

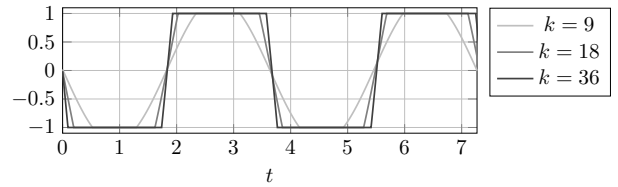
\begin{figure}[htbp]
    \centering
\begin{tikzpicture}[scale=0.8]
\begin{axis}[
        xlabel={$t$},
        grid=both,
        xmin=0, xmax=7.27,
        ymin=-1.1, ymax=1.1,
        no marks,
        legend pos=outer north east, 
        axis equal image
    ]
    
    \pgfplotstableread[col sep=comma]{figures/output.csv}\datatable
    \addplot[lightgray, thick] table[x=t, y=sat] {\datatable}; 
    \addlegendentry{$k=9$}
    \pgfplotstableread[col sep=comma]{figures/output3.csv}\datatablethree
    \addplot[gray, thick] table[x=t, y=sat] {\datatablethree}; 
    \addlegendentry{$k=18$}
    \pgfplotstableread[col sep=comma]{figures/output4.csv}\datatablefour
    \addplot[darkgray, thick] table[x=t, y=sat] {\datatablefour}; 
    \addlegendentry{$k=36$}
\end{axis}
\end{tikzpicture}
\caption{Oscillations in the feedback interconnection of Example~\ref{ex_1}.}
\label{fig:ex_1_oscillations}
\end{figure}
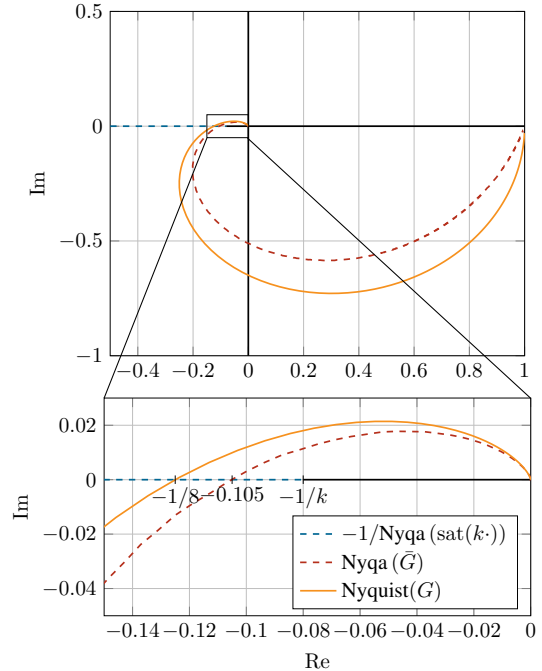
\begin{figure}[h]
\centering
\begin{tikzpicture}[scale=0.8]
\begin{groupplot}[
  group style={
    group size=1 by 2,
    y descriptions at=edge left,
    vertical sep=7mm
}]
    \nextgroupplot[
        ylabel={$\Im$},
        grid=both,
        xmin=-0.5, xmax=1.0,
        ymin=-1.0, ymax=0.5,
        no marks,
        legend pos=south east, 
    ]
    
    \addplot[black, thick] coordinates {(0, -1) (0, 0.5)};

    \addplot[black, thick] coordinates {(-0.08, 0) (1, 0)};

    \addplot[MidnightBlue, thick, dashed] coordinates {(-0.5, 0) (-0.08, 0)};

    \pgfplotstableread[col sep=comma]{example1adf.csv}\datatableone
    \addplot[BrickRed, dashed, thick] table[x=x, y=y] {\datatableone}; 

    \pgfplotstableread[col sep=comma]{example1nyq.csv}\datatabletwo
    \addplot[BurntOrange, thick] table[x=x2, y=y2] {\datatabletwo}; 

    \draw (-0.15,-0.05) rectangle (0,0.05);

    \nextgroupplot[
        xlabel={$\Re$},
        ylabel={$\Im$},
        grid=both,
        xmin=-0.15, xmax=0.0,
        ymin=-0.05, ymax=0.03,
        legend pos=south east, 
       yticklabel style={/pgf/number format/fixed},
       xticklabel style={/pgf/number format/fixed},
       scaled y ticks=false,
       legend cell align={left},
       width = \linewidth,
       height = 0.6\linewidth
    ]

    \addplot[MidnightBlue, thick, dashed, no marks] coordinates {(-0.5, 0) (-0.08, 0)};
    
    \addlegendentry{$-1/\nyq{\operatorname{sat}(k \cdot)}$}
    \addplot[black, thick, no marks, forget plot] coordinates {(-0.08, 0) (-0.0, 0)};
    \pgfplotstableread[col sep=comma]{example1adf.csv}\datatableone
    \addplot[BrickRed, dashed, no marks, thick] table[x=x, y=y] {\datatableone}; 
    \addlegendentry{$\nyq{\bar G}$};

    \pgfplotstableread[col sep=comma]{example1nyq.csv}\datatabletwo
    \addplot[BurntOrange, no marks, thick] table[x=x2, y=y2] {\datatabletwo}; 
    \addlegendentry{Nyquist$(G)$};

    \node [below] at (axis cs:  -0.08, 0) {$-1/k$};
    \node [below] at (axis cs:  -0.125, 0) {$-1/8$};
    \node [below] at (axis cs:  -0.105, 0) {$-0.105$};
    \addplot[only marks, mark=|] coordinates {
        (-0.08, 0)
        (-0.125, 0)
        (-0.105, 0)
    };
    \end{groupplot}
    

    \draw (1.6, 3.6) -- (-0.1, -0.7);
    \draw (2.27, 3.6) -- (6.95, -0.7);
\end{tikzpicture}
\caption{Nyquist diagram, amplitude Nyquist diagram and amplitude describing function plots for Example~\ref{ex_1}.}
\label{fig:example_1}
\end{figure}

\begin{figure}[htbp]
\centering
\begin{tikzpicture}[scale=0.8]
    \begin{axis}[
        xlabel={$k$},
        ylabel={$T$},
        grid=both,
        xmin=8, xmax=62, 
        ymin=3.62, ymax=3.69, 
                legend style={
            at={(1,0)}, 
            anchor=south east, 
            yshift=0.75cm, 
            xshift=-0.5cm
        },
        width=\linewidth,
        height=0.6\linewidth
    ]

    \addplot[only marks, mark=*] coordinates {
        (8,  10.898/3)
        (10, 10.910/3)
        (12, 10.925/3)
        (13, 10.941/3)
        (15, 10.960/3)
        (16, 10.966/3)
        (18, 10.979/3)
        (20, 10.992/3)
        (23, 11.000/3)
        (25, 11.009/3)
        (30, 11.016/3)
        (35, 11.023/3)
        (40, 11.020/3)
        (45, 11.025/3)
        (50, 11.021/3)
        (55, 11.022/3)
        (60, 11.022/3)
    };
    \addlegendentry{sim}

    \addplot[domain=8:62, samples=2, BurntOrange, thick] {3.632}; 
    \addlegendentry{df}

    \addplot[domain=8:62, samples=2, BrickRed, thick, dashed] {3.680}; 
    \addlegendentry{adf}

    \end{axis}
\end{tikzpicture}

\caption{Oscillation period in system \eqref{eq:example_1} for varying gain $k$ (sim), as well as periods predicted by the regular describing function (df) and amplitude describing function (adf).  Periods were measured by simulating the system with Runge-Kutta (4, 5) with a maximum step size of $0.0001$s in Simulink.}
\label{fig:example_1_period}
\end{figure}
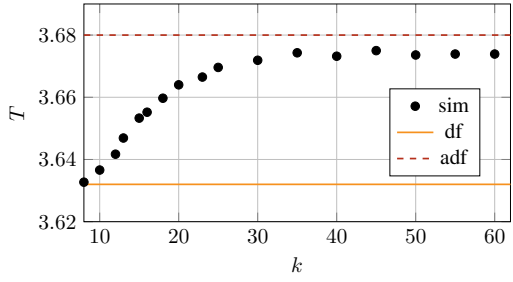

\begin{example}\label{ex:sat_delay}
    Let $D: \mathcal{B} \to \mathcal{B}$ be the amplitude-dependent delay defined in Example~\ref{ex:delay}, and consider the negative feedback interconnection shown in Figure~\ref{fig:adf_fb}, where $N: \mathcal{B} \to \mathcal{B}$ is defined as
    \begin{IEEEeqnarray*}{rCl}
         N(u)(t) &:=& \sat{D(u)}(t)\\
                 &=& \sat{u(t + \gamma(t))},\\
        \gamma(t) &=& \max_{\nu < t}|u(\nu)|,
    \end{IEEEeqnarray*}
    and $G(s) = k/(s+1)$.
    Given a square wave input $\alpha\sq$, the gain of this operator is $|\sat{\alpha}/\alpha|$, and the phase is $2\pi\alpha/T$.  For fixed $T$, the amplitude Nyquist diagram of this operator is therefore a segment of the unit circle ($|\alpha| \leq 1$) followed by a spiral towards the origin ($|\alpha| > 1$), illustrated in Figure~\ref{fig:ex_sat_delay}.


    We apply the amplitude describing function method to predict oscillations in the feedback interconnection.  As $\nyq{N}$ is $T$-dependent, we must draw multiple plots: we fix $T$, plot $-1/\nyq{N}$, check for an intersection with $\nyq{\bar G}$, where $G(s) = k/(s+1)$, and find the value of $T_0$ which corresponds to the intersection point on $\nyq{\bar G}$.  We then adjust $T$ and repeat, until $T$ (which determines $\nyq{N}$) matches $T_0$ (given by $\nyq{\bar G}$).  An example plot is shown in Figure~\ref{fig:ex_sat_delay_2}.  The method predicts oscillations for $k > 2.9$.  Experimentally, oscillations are present for $k > 2.2$, and at $k = 2.9$, the amplitude describing function method predicts a period of $T = 3.0$ and an amplitude (at the output of $G$) of $1.01$, whereas simulations give a period of $T = 6.6$ and an amplitude of $2.6$.  As $k$ increases, $N$ becomes increasingly low-pass (relative to the gain of $G$), and the predictions become more accurate.  For $k = 15$, the amplitude describing function method predicts a period of $T = 28.4$ and an amplitude of $13.47$, whereas the simulation gives a period of $T = 31.4$ and an amplitude of $15.0$.
\end{example}

\begin{figure}[htpb]
    \centering
\begin{tikzpicture}[scale=0.8]
\begin{axis}[
        xlabel={$\Re$},
        ylabel={$\Im$},
        grid=both,
        xmin=-0.5, xmax=1.1,
        ymin=-0.7, ymax=0.2,
        no marks,
        legend pos=south east, 
        axis equal image
    ]
    
    \addplot[black, thick] coordinates {(0, -0.7) (0, 0.2)};

    \addplot[black, thick] coordinates {(-0.5, 0) (1.1, 0)};

    \pgfplotstableread[col sep=comma]{figures/sat_delay_nyqa_2.csv}\datatabletwo
    \addplot[MidnightBlue, thick] table[x=x, y=y] {\datatabletwo}; 
\end{axis}
\end{tikzpicture}
\caption{Amplitude Nyquist diagram of the saturated, amplitude-dependent delay $N$ of
Example~\ref{ex:sat_delay}, for $T = 10$.}
\label{fig:ex_sat_delay}
\end{figure}
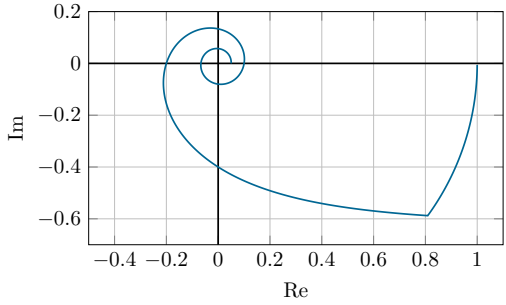

\begin{figure}[htbp]
    \centering
\begin{tikzpicture}[scale=0.8]
\begin{axis}[
        xlabel={$\Re$},
        ylabel={$\Im$},
        grid=both,
        xmin=-3, xmax=12.0,
        ymin=-9.0, ymax=1.0,
        no marks,
        legend pos=south east, 
        axis equal image
    ]
    
    \addplot[black, thick, forget plot] coordinates {(0, -9.0) (0, 1.0)};

    \addplot[black, thick, forget plot] coordinates {(-3, 0) (12, 0)};

    \pgfplotstableread[col sep=comma]{figures/lag_nyqa.csv}\datatableone
    \addplot[BrickRed, dashed, thick] table[x=x, y=y] {\datatableone}; 
    \addlegendentry{$\nyq{\bar G}$}

    \pgfplotstableread[col sep=comma]{figures/sat_delay_nyqa.csv}\datatabletwo
    \addplot[MidnightBlue, thick] table[x=x, y=y] {\datatabletwo}; 
    \addlegendentry{$-1/\nyq{N}$}
\end{axis}
\end{tikzpicture}
\caption{Amplitude Nyquist diagram and amplitude describing function plots for
Example~\ref{ex:sat_delay}, with $k = 11$. The amplitude Nyquist diagram is plotted
for $T = 20.4$, and the intersection occurs at the $T = 20.4$ point on the amplitude
describing function, predicting an oscillation of this period.}
\label{fig:ex_sat_delay_2}
\end{figure}
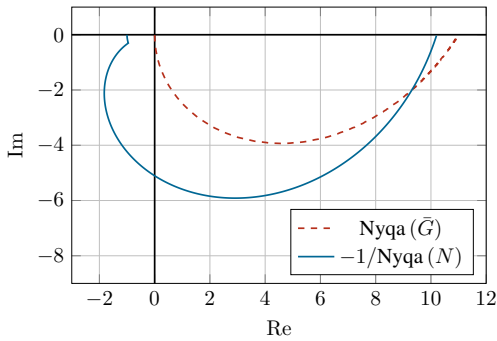

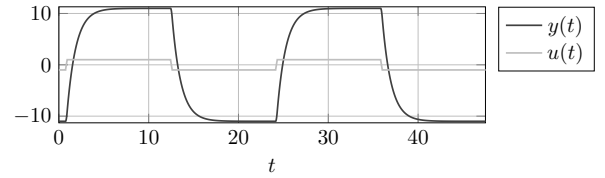
\begin{figure}[h]
    \centering
\begin{tikzpicture}[scale=0.8]
\begin{axis}[
        xlabel={$t$},
        grid=both,
        xmin=0, xmax=47.5,
        ymin=-11.3, ymax=11.3,
        no marks,
        legend pos=outer north east, 
        height=3.5cm,width=\linewidth
    ]
    
    \pgfplotstableread[col sep=comma]{figures/output2.csv}\datatable
    \addplot[darkgray, thick] table[x=t, y=lag] {\datatable}; 
    \addlegendentry{$y(t)$}
    \addplot[lightgray, thick] table[x=t, y=sat] {\datatable}; 
    \addlegendentry{$u(t)$}
\end{axis}
\end{tikzpicture}
\caption{Oscillations in the feedback interconnection of Example~\ref{ex:sat_delay} for $k = 11$.}
\label{fig:ex_2_oscillations}
\end{figure}

\section{Conclusions and future directions}\label{sec:conclusions}

We have introduced a version of the describing function method where sinusoids are
replaced by square waves.  Static nonlinearities become the ``easy'' components,
mapping square waves to square waves and allowing an analog of frequency response
analysis, which we call amplitude response.  The output of an LTI system to a square wave, however, is not square, and we approximate it by a square wave to give a square version of the describing function.  The describing function method of predicting oscillations is generalized to this square wave setting, and predicts the existence and properties of approximately square oscillations better than the classical describing function method.  

Square waves are only one class of signals, and an area for future research is other classes of signals, and systems which preserve them, for which a similar analysis can be accomplished.  A second area for future research is giving rigorous bounds on the error of the amplitude describing function approximation, in a manner similar to work on the classical describing function \cite{Bergen1971, Mees1975}.

\begin{appendices}
    \section{Computing the transform}
In this appendix, we address the computation of the amplitude describing function.  Proposition~\ref{thm:transform_simple} simplifies the characterization of Theorem~\ref{thm:square_transform},  reducing the number of inner products that must be evaluated to compute the transform.  Proposition~\ref{prop:pmp} then gives a class of systems for which this characterization further simplifies and the phase shift $\tau$ may be computed by searching for points which are visited periodically at period $T/2$.
\begin{proposition}\label{thm:transform_simple}
    Let $N: \LT \to \LT$ map pointwise bounded inputs to pointwise bounded outputs, and let $\alpha \in \Ro$, $T > 0$ and $y = N(\alpha \sq)$.  Then 
    if there exists $\tau \in [0, T/2]$ such that
    \begin{IEEEeqnarray*}{rCl}
       \tau &=& \argmin\limits_{\tilde \tau \text{ s.t. } y(\tilde\tau) = y(\tilde\tau + T/2)}
           -\ip{y}{\sq_{\tilde\tau}}^2
    \end{IEEEeqnarray*}
    $\phi = \frac{2\pi}{T}\tau$ solves~\eqref{eq:opt}, and otherwise if $\tau$ is a solution to 
    \begin{IEEEeqnarray*}{rCl}
           -\int_0^{\frac{T}{2\pi}\tau} y(t)\dd{t} + \int_0^{\frac{T}{2\pi}\tau + \frac{T}{2}} y(t) \dd{t}  =
     \frac{1}{2}\int_0^T y(t) \dd{t}
    \end{IEEEeqnarray*}   
    $\phi = \frac{2\pi}{T}\tau$ solves  \eqref{eq:opt}.
\end{proposition}

\begin{proof}
    From Theorem~\ref{thm:square_transform}, we have 
    \begin{IEEEeqnarray*}{rCl}
    \tau &=& \argmin\limits_{\tilde{\tau} \in [0, T/2]}
    -\ip{y}{\sq_{\tilde\tau}}^2.
    \end{IEEEeqnarray*}
    Boundedness of $y$ implies boundedness of $\ip{y}{\sq_{\tilde \tau}}$, so this minimization admits at least one solution.
    Computing the inner product explicitly for $\tilde \tau \in [0, T/2]$ gives
    \begin{IEEEeqnarray*}{rCl}
    \scriptstyle
        &&\ip{y}{\sq_{\tilde{\tau}}}^2\\ &=&\frac{1}{T^2}\left( \int_0^{\tilde{\tau}} -y(t) \dd{t} +
        \int_{\tilde{\tau}}^{{\tilde{\tau}} + \frac{T}{2}} y(t) \dd{t}
                            + \int_{{\tilde{\tau}} + \frac{T}{2}}^T -y(t) \dd{t}\right)^2\\
                            &=& \frac{1}{T^2}\left( -2\int_0^{\tilde{\tau}} y(t)\dd{t} +
                                2\int_0^{{\tilde{\tau}} + \frac{T}{2}} y(t) \dd{t}
                            - \int_0^T y(t) \dd{t}\right)^2,
    \end{IEEEeqnarray*}
    from which we have
    \begin{IEEEeqnarray*}{rCl}
    -\td{}{{\tilde{\tau}}} \ip{y}{\sq_{\tilde{\tau}}}^2 &=&
    -2\ip{y}{\sq_{\tilde{\tau}}}\td{}{{\tilde{\tau}}}\ip{y}{\sq_{\tilde{\tau}}}\\
                                    &=& \frac{4}{T^2}\ip{y}{\sq_{\tilde{\tau}}}
                                    \left(y({\tilde{\tau}}) - y\left({\tilde{\tau}} +
                                    \frac{T}{2}\right)\right).
    \end{IEEEeqnarray*}
    
    This derivative is zero in two cases.  Firstly, when
    \begin{IEEEeqnarray*}{rCl}
    -\int_0^{\tilde\tau} y(t)\dd{t} + \int_0^{\tilde\tau + \frac{T}{2}} y(t) \dd{t}  &=&
      \frac{1}{2}\int_0^T y(t) \dd{t},
    \end{IEEEeqnarray*}
    in which case $\ip{y}{\sq_{\tilde\tau}} = 0$, and secondly when $y(\tilde\tau) = y(\tilde\tau + T/2)$.  These correspond to the two cases of the theorem statement (noting that $-\ip{y}{\sq_{\tilde\tau}}^2 \leq 0$).
\end{proof}

We have tacitly avoided the question of whether the minimization in Definition~\ref{def:adf}
has a 
unique solution.  The following proposition gives a class of systems for which this
is the case. 
A periodic signal is called \emph{periodic monotone} if it intersects any horizontal line $y(t) = \gamma$ at most twice in a period.  A system $G$ is called \emph{periodic monotone preserving} \cite{Ruscheweyh1992, Grussler2023, Grussler2025} if it is periodic preserving and, given a periodically monotone input $u$, $y = G(u)$ is periodically monotone.
\begin{proposition}\label{prop:pmp}
    Suppose $G(s) = H(s)/(s+p)$, where $p > 0$ and $H(s)$ is stable and periodic monotone preserving.  Then
    there is a unique $\tau$ for which $y(\tau) = y(\tau + T/2)$.
\end{proposition}
\begin{proof}
    Let $u(t) = \sq(t)$.  Then applying $1/(s+p)$ to $u(t)$ gives a steady-state which is continuous and periodic monotone.  It follows that $y = G(u)$ is continuous (since $H$ is stable), $T$-periodic and periodic monotone.  Equivalently \cite{Ruscheweyh1992}, there exist $t_1, t_2 \in \R$, $t_2 > t_1$, such that $y(t)$ is increasing for $t_1 \leq t \leq t_2$ and decreasing for $t_2 \leq t \leq t_1 + T$.  Consider one period of $y$ from $t_1$ to $t_1$ to $t_1 + T$, and define $d(\gamma)$ to be the distance between the two intersection points of this period with the horizontal line at height $\gamma$ if there are two intersection points, and zero if there are one or zero intersection points.
    
    \begin{center}
    \includegraphics[scale=0.8]{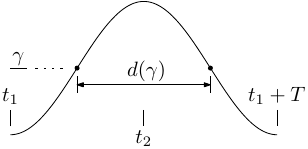} 
    \end{center}
    
    As $t_1$ $t_1 + T$ are at the minimum of $y$ by construction, $d(y(t_1)) = T$.  Similarly, $t_2$ is at the maximum, and $d(y(t_2)) = 0$. For $\gamma \in (y(t_1), y(t_2))$, as $y$ is continuous, $d(\gamma)$ is continuous and monotonically decreasing, therefore there exists a unique $\gamma \in (y(t_1), y(t_2))$ such that $d(\gamma) = T/2$, which concludes the proof.
\end{proof}
\end{appendices}

\bibliographystyle{IEEEtran}
\bibliography{describing_functions}

\end{document}